\documentclass[11pt]{article}

\usepackage{graphicx} 
\usepackage{amsfonts}
\usepackage{amssymb}
\usepackage{amsmath}
\usepackage{bbm}
\usepackage{enumitem}

\usepackage[usenames,svgnames]{xcolor}
\usepackage{url}
\usepackage{umoline}

\usepackage{hyperref}
\hypersetup{
     colorlinks=true,       		
     linkcolor=Salmon,          	
     citecolor=blue,            
     filecolor=blue,      		
     urlcolor=cyan,           	
 }

\newtheorem{theorem}{Theorem}[section]

\newtheorem{deff}[theorem]{Definition}

\newtheorem{lem}[theorem]{Lemma}

\newtheorem{corol}[theorem]{Corollary}

\newcommand{\ignore}[1]{{}}

\newcommand{\Eq}[1]{Eq.~(\ref{#1})}
\newcommand{\Fig}[1]{Fig.~\ref{#1}}
\newcommand{\Def}[1]{Def.~\ref{#1}}
\newcommand{\Lem}[1]{Lemma~\ref{#1}}

\newcommand{\Sec}[1]{Sec.~\ref{#1}}
\newcommand{\Ref}[1]{Ref.~\cite{#1}}
\newcommand{\Thm}[1]{Theorem~\ref{#1}}

\newcommand{\Id}{{\mathbbm{1}}}

\newcommand{\norm}[1]{{\| #1 \|}}  
  
\newcommand{\ket}[1]{{ |{#1} \rangle }}  
\newcommand{\bra}[1]{{ \langle {#1} | }}
\newcommand{\braket}[2]{{ \langle {#1} | {#2} \rangle}}

\newcommand{\Oorderof}{\mathcal{O}}
\newcommand{\orderof}[1]{\Oorderof(#1)} 
 
\newcommand{\EqDef}{\stackrel{\mathrm{def}}{=}}

\newcommand{\eps}{\epsilon}  

 \newcommand{\qedsymb}{\hfill{\rule{2mm}{2mm}}}  
 \newenvironment{proof}[1][]{\begin{trivlist}  
 \item[\hspace{\labelsep}{\bf\noindent Proof#1:\/}]}
 {\qedsymb\end{trivlist}}

\newcommand{\Egs}{\epsilon_0}

\newcommand{\Lin}{L}
\newcommand{\Lout}{L^c}
\newcommand{\Lb}{\partial L}
\newcommand{\Lc}{\overline{L}}
\newcommand{\Hin}{H_L}
\newcommand{\Hout}{H_{L^c}}
\newcommand{\Hb}{H_\partial}
\newcommand{\tH}{\tilde{H}}
\newcommand{\tP}{\tilde{P}}
\newcommand{\tHin}{\tilde{H}_L}
\newcommand{\Ein}{\eps_0(\Lin)}
\newcommand{\Eout}{\eps_0(\Lout)}

\newcommand{\teps}{\tilde{\eps}}
\newcommand{\tEgs}{\teps_0}
\newcommand{\tPi}{\tilde{\Pi}}

\begin{document}

\title{Connecting global and local energy
  distributions in quantum spin models on a lattice}

\author{Itai Arad\footnote{Centre for Quantum Technologies, 
  National University of Singapore, Singapore.}\ \and 
  Tomotaka Kuwahara\footnote{Department of Physics, The University
  of Tokyo, Japan}\ \and
  Zeph Landau\footnote{UC Berkeley, California, USA.}}
\date{\today}

\maketitle

\begin{abstract}
  Local interactions in many-body quantum systems are generally
  non-commuting and consequently the Hamiltonian of a local region
  cannot be measured simultaneously with the global Hamiltonian.
  The connection between the probability distributions of
  measurement outcomes of the local and global Hamiltonians will
  depend on the angles between the diagonalizing bases of these two
  Hamiltonians. In this paper we characterize the relation between
  these two distributions. On one hand, we upperbound the
  probability of measuring an energy $\tau$ in a local region, if
  the global system is in a superposition of eigenstates with
  energies $\eps<\tau$. On the other hand, we bound the probability
  of measuring a global energy $\eps$ in a bipartite system that is
  in a tensor product of eigenstates of its two subsystems. Very
  roughly, we show that due to the local nature of the governing
  interactions, these distributions are identical to what one
  encounters in the commuting case, up to some exponentially small
  corrections. Finally, we use these bounds to study the spectrum of
  a locally truncated Hamiltonian, in which the energies of a
  contiguous region have been truncated above some threshold energy
  $\tau$. We show that the lower part of the spectrum of this
  Hamiltonian is exponentially close to that of the original
  Hamiltonian. A restricted version of this result in 1D was a
  central building block in a recent improvement of the 1D area-law.
\end{abstract}

\section{Introduction}
\label{sec:introduction}

One of the most striking features of strongly interacting many-body
quantum systems is that despite their overwhelming complexity, they
often exhibit strong properties of locality, which make them
accessible analytically. A prominent example of such property is a
bound on the speed at which the effect of a local perturbation
spreads in lattice spin models, the well-known \emph{Lieb-Robinson
bound}~\cite{ref:LR-bound72, ref:LSM-Hastings04, 
ref:Nachtergaele2006-LR, ref:Cheneau2012-Light, 
ref:Richerme2014-NonLocal}. This bound is the backbone of numerous
fundamental results in quantum many-body systems: the
Lieb-Schultz-Mattis theorem~\cite{ref:LSM-Hastings04,
ref:Nachtergaele2007-LSM} the exponential decay of
correlations~\cite{ref:Hastings2004-Markov, ref:Hastings2006-ExpDec,
ref:Nachtergaele2006-LR}, the 1D
area-law~\cite{ref:Hastings2007-1D}, the complexity of quantum
simulations~\cite{ref:Osborne2006-Efficient,
ref:Osborne2007-Adiabatic, ref:Hastings2009-Adiabatic}, the
stability of topological order to
perturbation~\cite{ref:Hastings2005-TO-Stability,
ref:Bravyi2006-Topo-Order, ref:Bravyi2010-TO-Stability,
ref:Bravyi2011-TO-Stability}, the quantization of the hall
conductance~\cite{ref:Hastings2014-Hall}, and so on. 

In the present paper we expose a new kind of locality property in 
many-body quantum spin systems, which is related to their
Hamiltonian eigenstates. The main question we ask is how the
eigenstates of the global system's Hamiltonian are related to those
of a subsystem. Specifically, consider a many-body quantum spin
system on a lattice $\Lambda$ that is described by a local
Hamiltonian $H=\sum_{X\subset \Lambda} h_X$. We assume that every
interaction involves at most $k$ particles, and the the total
strength of all interactions that involve a particle is finite. This
includes almost all interesting spin models with short-range
interactions such as the the XY model~\cite{ref:XY}, the Heisenberg
model~\cite{ref:Haldane1983-HeisenbergI,ref:Haldane1983-HeisenbergII}
and the AKLT model~\cite{ref:AKLT}, as well as some models with
long-range interactions such as the Lipkin-Meshcov-Glick
model~\cite{ref:LMG}. We consider a region $L\subset
\Lambda$ in the lattice with $L^c=\Lambda\setminus L$ its
complement, and decompose $H$ into 3 parts, $H=\Hin+\Hout+\Hb$:
$\Hin$ acting on the particles in $L$, $\Hout$ on those in $L^c$ and
$\Hb$ contains interactions involving both particles in $L$ and
$L^c$. We then ask two basic questions (Figs.~\ref{fig:settings} and \ref{fig:dists-illustration}):
\begin{enumerate}
  \item Given an eigenstate $\ket{\psi_\eps}$ of $H$,
    how does its expansion in terms of eigenstates of $\Hin$ look
    like? More generally, if $I_1\EqDef [\eps_a, \eps_b]$,
    $I_2\EqDef[\eps_c, \eps_d]$ are ranges of energies and
    $\Pi_{I_1}$ is the projector into the subspace
    of eigenstates of $H$ with energies in $I_1$ and similarly
    $P_{I_2}$ is defined with respect to $\Hin$, then what can we say
    about the overlap $\norm{P_{I_2}\Pi_{I_1}}$, where
    $\norm{\cdots}$ denotes the operator norm?
    
  \item Similarly, if $\ket{\psi_\eps}$ is an eigenstate
    of $\Hin + \Hout = H-\Hb$, how does its expansion in terms of
    eigenstates of $H$ look like? Or, more generally, if
    $I_1\EqDef [\eps_a, \eps_b]$,
    $I_2\EqDef[\eps_c, \eps_d]$ are ranges of energies, 
    $\Pi_{I_1}$ is defined as above and $Q_{I_2}$ is defined with
    respect to $\Hin+\Hout$,  then what can we say
    about the overlap $\norm{\Pi_{I_1}Q_{I_2}}$?
\end{enumerate}
As we shall see, to answer these questions we shall need to answer
the following more basic question:
\setlist[enumerate,1]{start=0}
\begin{enumerate}
  \item Given an eigenstate $\ket{\psi_\eps}$ of $H$, and an 
    operator $A$ that acts on region of the lattice $L\subset
    \Lambda$,  how does the expansion of $A\ket{\psi_\eps}$ in
    terms of eigenstates of $H$ looks like? In other words, what can
    we say about the transition probability
    $|\bra{\psi_{\eps'}}A\ket{\psi_\eps}|^2$ for another
    eigenstate $\ket{\psi_{\eps'}}$?
\end{enumerate}
\setlist[enumerate,1]{start=1}

\begin{figure}
  \begin{center}
    \includegraphics[scale=1.0]{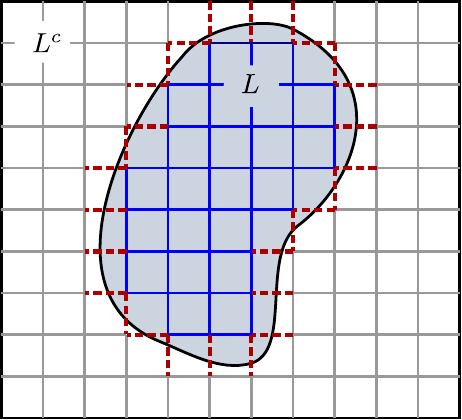}
  \end{center}
  \caption{\label{fig:settings} An illustration of a quantum spin
    model on a lattice (see Section~\ref{sec:settings} in detail).
    Edges denote interactions between the particles that sit on the
    vertices. A contiguous region $\Lin$ in the lattice decomposes
    the particles into two sets, those inside $\Lin$ and outside of
    it. This, in turn, defines a decomposition of $H$ into 3 parts:
    $\Hin$, which is made of all the interactions among particles
    inside $\Lin$ (blue edges), $\Hout$, which is made of all the
    interactions among particles inside $\Lout$ (gray edges), and
    $\Hb$ which is made of all the remaining interactions denoted by
    red dotted edges. } 
\end{figure}

\begin{figure}
  \begin{center}
    \includegraphics[scale=1.5]{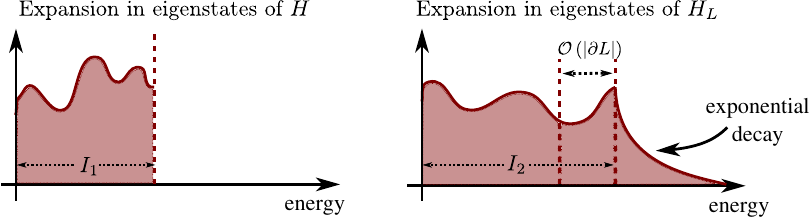}
  \end{center}
  \caption{\label{fig:dists-illustration} The connection between the
  weight distribution of eigenstates of $H$ of a given state to the
  distribution of the \emph{same} state but with respect to the
  eigenstates of $\Hin$. If the eigenstates of $H$ are supported in
  a region $I_1$, then up to an exponentially small tail, the
  eigenstates of $\Hin$ will be supported on a segment $I_2$, which
  is only larger than $I_1$ by a constant of $\orderof{|\Lb|}$. See
  \Thm{thm:dist} for a precise statement.}
\end{figure}

%
%
Our results require some care to be stated precisely, but the
summary is that \emph{due to the locality of the underlying
Hamiltonians}, up to some exponentially small corrections, the
system behaves largely as if the underlying Hamiltonians are
commuting; states that are well localized with respect to one
Hamiltonian are well localized with respect to the other. This fits
into a family many results in many-body quantum systems where the
general case resembles the behavior of models with commuting
interactions. Indeed, central to the proof is the assumption that
every local interaction term in $H$ commutes with all but a constant
number of terms.\footnote{The precise requirement, given in
\Sec{sec:settings}, is actually less stringent. We only require that
for any local interaction term, the sum of the norms of all other
terms that do not commute with it is bounded by an $\orderof{1}$
constant.} Other examples of this phenomenon include (but are not
restricted to) many of the Lieb-Robinson bound results mentioned
above such as the exponential decay of correlations in gapped
ground states, and the \emph{area-law} behavior often observed in
gapped systems \cite{ref:AL-rev} and rigorously proved for 1D gapped
systems. We note that all these results trivially hold in the
commuting case.

%
%

Apart from being very natural, the motivation behind the above
questions is three-fold. The initial motivation was a recent proof
of the 1D area-law in which it was needed to construct a Hamiltonian
whose local spectrum is very close to the original 1D Hamiltonian,
but its norm over some large parts of the system is truncated (see
\Def{def:Ht} for a precise statement). As we will see in the proof of
\Thm{thm:lowspec}, by answering questions (0-2), we are able to
construct such Hamiltonians in any lattice dimension. It is
reasonable to believe that using our techniques other interesting
constructions can be done. 

The second place where our results may prove useful is in the
analysis of many-body quantum dynamics of closed
systems~\cite{ref:Polkovnikov2011-Quench}. There, the dynamics is
governed by the Schr\"odinger equation
$\ket{\psi(t)}=e^{-iHt}\ket{\psi(0)}$, and can be calculated from 
the decomposition of $\ket{\psi(0)}$ in terms of eigenstates of $H$.
Our results (specifically, question 0) then can be useful for states
like $\ket{\psi(0)} = A\ket{\eps}$, where $\ket{\eps}$ is an
eigenstate of the Hamiltonian and $A$ is some local operator. This
dynamic is particularly relevant for calculating the spectral
function of lattice models~\cite{ref:Spectral-I, ref:Spectral-II},
as well as for understanding quantum quenches
\cite{ref:Polkovnikov2011-Quench}.

Finally, our results, particularly questions (1-2), seem highly
relevant for the question of thermalization of closed quantum
systems (see \Ref{ref:Polkovnikov2011-Quench,
ref:Eisert2015-Thermal} and references within) and related subjects
as the Eigenstate Thermalization Hypothesis
(ETH)~\cite{ref:Deutsch1991-QStat, ref:Srednicki1994-ETH,
ref:Tasaki1998-ETH, ref:Rigol2008-ETH, ref:Rigol2012-altETH}, the first law of thermodynamics~\cite{PhysRevE.92.032115} and relaxation process in periodically driven systems~\cite{ref:floquet}.

{~}

\noindent\textbf{Relation to previous work:}\\ Surprisingly, despite
the natural character of our main questions, we are not aware of
previous works that aim directly at them. Nevertheless, there are
some ``near by'' results. Perhaps the most relevant result is the
so-called ``local diagonality of energy eigenstates'', which is
proved by Muller et~al in \Ref{ref:Mueller2013-Thermal}. There, the
authors use the Lieb-Robinson bound to prove a slightly weaker
version of one of the necessary conditions for ETH: that the reduced
density matrix of a global energy eigenstate $\ket{\eps}$ over some
region $L$ is (almost) diagonal in the local energy eigenbasis.

{~}

\noindent\textbf{Organization of the paper:}\\ In \Sec{sec:results}
we provide a self-contained statement of our main results, together
with a description of spin systems to which they apply. In
Sections~\ref{sec:expE}--\ref{sec:Ht} we provide the full
proofs of these theorems. In \Sec{sec:summary} we conclude with a
summary and some open questions and future directions.

\section{Statement of results}
\label{sec:results}

\subsection{Notation and general setup}
\label{sec:settings}

We consider a quantum system of $N$ quantum particles (spins) of
local dimension $d$ that are located on the vertices of some $D$
dimensional lattice $\Lambda$. We think of $N$ as a large number,
but we are not assuming the thermodynamic limit. The interaction
between the particles is governed by a generic few-body Hamiltonian $H$:
\begin{align}
  H &= \sum_{X\subset \Lambda} h_X \,, 
\end{align}
where $X\subset\Lambda$ are subsets of particles, and we assume that
\emph{$h_X=0$ for $|X|>k$}, that is, all interactions involve at
most $k$ particles. 
By shifting and rescaling the local terms, we can always
pass to dimensionless units in which
\begin{align}
\label{eq:hX}
h_X  \ge 0
\end{align}
i.e., $h_X$ are semi-positive definite.  We \emph{do not} assume
that $h_X$ involves only neighboring particles on the lattice.
Instead, we impose the weaker condition that there exists some
constant $g=\orderof{1}$ such for every particle $i$, 
\begin{align}
\label{def:g}
  \sum_{X: i\in X} \norm{h_X}\le g \,.
\end{align}
That is, the sum of norms of all interactions that involve particle
$i$ is bounded by $g$.  All nearest-neighbors systems on a
$D$-dimensional square lattice satisfy this property with, say,
$g\le (2D)^{k-1}$. In addition, it also satisfied by some some
models with long-range interactions as the Lipkin-Meshcov-Glick
model~\cite{ref:LMG} (i.e., the infinite range XY model). Finally, a
constant that we shall often use is
\begin{align}
\label{def:lambda}
  \lambda &\EqDef\frac{1}{2gk} \,.
\end{align}
    
Although we do not treat fermionic systems explicitly, our discussions can be also applied to various class of local
fermionic systems because they can be mapped into local spin systems \cite{PhysRevLett.95.176407,1742-5468-2005-09-P09012}. 
On the other hand, for bosonic systems, we cannot generally assume the inequality~\eqref{def:g} because an arbitrary number of boson can be in the same site and the one-site energy is not upperbounded. 
Therefore, in order to apply our discussion to bosonic systems, 
we need additional assumptions such as the hard-core boson.

We denote the energy levels of the system (the eigenvalues of $H$)
by $0\le \eps_0\le \eps_1\le\eps_2 \ldots$, and their corresponding
eigenvectors by $\ket{\psi_0}, \ket{\psi_1}, \ket{\psi_2},\ldots$.
Notice that $\eps_0\ge 0$ since we assume that every $h_X$ is
non-negative.

Throughout, we let $\Lin$ denote a subset of the particles, and
$\Lout$ the complementary subset.  We usually envision the case
where the particles of $L$ are sitting in a contiguous region of the
system, but this is not a requirement. Given a subset $L$, we can
partition the $h_X$ terms in $H$ into three subsets $E_{\Lin},
E_{\Lout}, E_{\Lb}$ depending on whether their non-trivial action is
within $L$, within $L^c$, or involving both particles in $\Lin$ and
$\Lout$. For each subset we then define the corresponding
Hamiltonian
\begin{align}
  \label{def:Hamiltonians}
  \Hin &\EqDef \sum_{X\in E_{\Lin}}h_X \ , & 
  \Hb &\EqDef \sum_{X\in E_{\Lb}}h_X \ ,   &
  \Hout &\EqDef \sum_{X\in E_{\Lout}}h_X \,,
\end{align}
so that
\begin{align}
  H = \Hin + \Hb + \Hout \,.
\end{align}
This decomposition is illustrated in \Fig{fig:settings}.  We denote
the energy levels of $\Hin$ by $\eps_0(\Lin)\le\eps_1(\Lin)\le
\ldots$, and the energy levels of $\Hout$ by
$\eps_0(\Lout)\le\eps_1(\Lout)\le\ldots$.

By a slight abuse of notation, we define
\begin{align}
\label{def:sizes}
  |\Lin| &\EqDef \sum_{X\in E_{\Lin}} \norm{h_X} \ , &
  |\Lb| &\EqDef \sum_{X\in E_{\Lb}} \norm{h_X} \ , &
  |\Lout| &\EqDef \sum_{X\in E_{\Lout}} \norm{h_X} \,.
\end{align}
Notice that when each $h_X$ has exactly norm 1 and is defined on
exactly $k$ particles, and in addition every particle participates
in exactly $g$ interactions then the number of particles in $\Lin$
is indeed $\orderof{\frac{k}{g}|L|}$. Finally, we define
$|\Lc|\EqDef |\Lin| + |\Lb|$.

\subsection{Main results}

We begin with \Thm{thm:expE}, the backbone of \Thm{thm:dist} and
\Thm{thm:product}. It bounds the effect of an arbitrary operator $A$
on a superposition of eigenstates of $H$. Specifically, we assume
that we are given a state that is a superposition of eigenstates of
$H$ with energies in $[0,\eps]$ and then some operator $A$ (say, a
unitary transformation) is applied. The resultant state, of course,
can contain eigenstates of $H$ outside $[0,\eps]$, with energies
greater than $\eps$. Classically, if the norm of every local term is
at most one, and we apply a transformation on a region $L$, the total
energy can change by at most the number of local Hamiltonian
terms it touches, i.e., by $|\Lc|$.  In the quantum case a similar
thing holds, up to some exponentially small corrections: the energy
distribution is concentrated on the interval $[0,\eps+|\Lc|]$, and
the weight of eigenstates with energy above $\eps+|\Lc|$ is
exponentially suppressed. When $A$ commutes with $\Hb$, the
concentration is on the tighter interval $[0,\eps+|\Lb|]$. The proof
of this theorem is based on an unpublished result by M.~B.~Hastings
which proved that for any operator $A$ supported on $L$,
$\norm{\Pi_{[\eps',\infty)}\, A \, \Pi_{[0,\eps]}} \le
e^{-\orderof{\lambda(\eps'-\eps)/|L|}}\cdot \norm{A}$.

\begin{theorem}
\label{thm:expE} Let $\Pi_{[\eps',\infty)}$ and $\Pi_{[0,\eps]}$ be
  projectors onto the subspaces of energies of $H$ that are $\ge
  \eps'$ and $\le \eps$ respectively. For an operator $A$, let $E_A$
  be a subset of interaction terms such that $[H,A]=\sum_{X\in E_A}
  [h_X,A]$, and let $R\EqDef\sum_{X\in E_A} \norm{h_X}$.  Then
  \begin{align}
  \label{eq:general-expE}
    \norm{\Pi_{[\eps',\infty)}\, A \, \Pi_{[0,\eps]}}
      \le  \norm{A}\cdot \exp\Big\{ -\frac{1}{gk}\big[
         \eps'-\eps-R(1+\ln\frac{\eps-\eps'}{R})\big]\Big\} 
       \le \norm{A}\cdot e^{-\lambda(\eps'-\eps-2R)} \,,
  \end{align}
  where $\lambda \EqDef \frac{1}{2gk}$ was defined in
  \Eq{def:lambda}.\\ \noindent\textbf{Note:} When $A$ is supported
  on a subset of $\Lin$ particles, we can set $R=|\Lc| = |\Lin| +
  |\Lb|$. If in addition $[A,\Hin]=0$, we may set $R=|\Lb|$.
\end{theorem}

The proof uses similar techniques to those that are used in the
proof of the Lieb-Robinson bound\cite{ref:LR-bound72,
ref:LSM-Hastings04, ref:Nachtergaele2006-LR}. In particular, we
exploit the local nature of $H$ using the Hadamard formula $e^{s
H}Ae^{-s H}= A + s[H,A] + \frac{s^2}{2!}[H,[H,A]] + \ldots$. The
fact that $H$ is a sum of local terms implies that the commutators
on the RHS contain a finite number of terms, and their norm can
therefore be bounded.

We use \Thm{thm:expE} in the proofs of our main results that relate
the shape of the energy distributions of $H$ to that of $\Hin$ and
$\Hin + \Hout$ for parts of the system.  \Thm{thm:dist} shows that
states that are superposition of eigenstates of $H$ with energies
in $[0,\eps]$ can be expanded in eigenstates of $\Hin$ with energies
in $[0,\eps-\eps_0 + \eps_0(L) + 3|\Lb|]$, plus some eigenstates
outside that range with exponentially small weights. The upperbound 
$\eps-\eps_0 + \eps_0(L) +3|\Lb|$ has the following intuitive
interpretation. It
consists of two parts: $\eps-\eps_0 + \eps_0(L)$ maps an energy
excitation in $H$ to the same excitation in $\Hin$ by shifting the
ground energies $\eps_0\to \eps_0(L)$. The second part, $3|\Lb|$,
corresponds to a widening of the range due to the boundary
interactions.
\begin{theorem}
\label{thm:dist} Let $P_{[\tau,\infty)}$ denote the projection onto
  the subspace of energies of $H_L$ which are $\ge \tau$, and let
  $\Pi_{[0,\eps]}$ denote the projection onto the subspace of
  energies $H$ that are $\le \eps$. Then
  \begin{align}
  \label{eq:EnergyDist}
    \norm{P_{[\tau,\infty)}\Pi_{[0,\eps]}} \le 
    \frac{2}{\lambda^{1/2}}
      \cdot e^{-\lambda(\Delta \tau - \Delta \eps - 3|\Lb|)} \,,
  \end{align}
  where $\Delta\tau \EqDef \tau-\eps_0(\Lin)$ and $\Delta \eps
  \EqDef \eps-\eps_0$, and $\eps_0(\Lin)$ and $\eps_0$ are the
  ground energies of $H_L$ and $H$ respectively. 
\end{theorem}

Our next result, \Thm{thm:product}, addresses the question of how
the shapes of the energy distributions for the whole system compares
to that of two isolated complementary parts. In other
words, how does the interaction between the two complementary parts
changes the energy distribution. Specifically, it shows that any
superposition of eigenstates of $\Hin+\Hout = H-\Hb$ with energies in
$[a,b]$ can be expanded as a superposition of eigenstate of $H$ in a
larger region $[a-3|\Lb|,b+3|\Lb|]$, plus some exponentially
small contributions from outside that region.
\begin{theorem}
  \label{thm:product} 
  Let $\Lin$ be a subset of particles and let $H=\Hin+\Hb+\Hout$ be
  its corresponding decomposition of $H$. Let $Q_I$ be the projector
  into the subspace of eigenstates of $\Hin+\Hout$ with energies in
  the range $I$, and let $\Pi_I$ be the corresponding projector of
  $H$. Then for any energy scales $\tau>\eps>0$,
  \begin{align}
  \label{eq:t-gt-eps}
    \norm{\Pi_{[0,\eps]}Q_{[\tau,\infty)}}
      \le \frac{2}{\lambda^{1/2}} e^{-\lambda(\tau-\eps - 3|\Lb|)} \,,
  \end{align}
  and for $\eps>\tau>0$, 
  \begin{align}
  \label{eq:eps-gt-t}
    \norm{\Pi_{[\eps, \infty)}Q_{[0,\tau]}}
      \le \frac{2}{\lambda^{1/2}} e^{-\lambda(\eps-\tau-3|\Lb|)} \,.
  \end{align}
\end{theorem}
The proof follows the same
lines as \Thm{thm:product} with some small modifications.



An immediate  corollary of this theorem is the following bound on the
energy distribution of a product state:
\begin{corol}[Energy distribution of a product state]
  Under the same conditions of \Thm{thm:product}, let
  $\ket{\psi_{\Lin}}$ be an eigenstate of $\Hin$ with energy
  $\eps_{\Lin}$ defined on the Hilbert space supported by the
  particles of $\Lin$, and let $\ket{\psi_{\Lout}}$ be an eigenstate
  of $\Hout$ with energy $\eps_{\Lout}$ defined on the
  particles of
  $\Lout$, and set $\ket{\psi}\EqDef
  \ket{\psi_{\Lin}}\otimes\ket{\psi_{\Lout}}$. Then for any
  eigenstate $\ket{\eps}$ of $H$ with energy $\eps$, 
  \begin{align}
    |\braket{\eps}{\psi}| 
      \le \frac{2}{\lambda^{1/2}} 
        e^{-\lambda(|\eps_{\Lin} + \eps_{\Lout}-\eps| - 3|\Lb|)} \,.
  \end{align}
\end{corol}

We now turn to our final result, which is one possible application
of our main results.  When studying the physics of a quantum lattice
spin system, it is often desirable to approximate the Hamiltonian
$H$ by a new Hamiltonian $\tH$ that is identical to $H$ in some
local region, but nevertheless has a bounded norm that does not
scale extensively with the system size. This restriction on the norm
is necessary, for example, when one wants to approximate the
groundspace projector using a low-degree polynomial of $H$. For a
polynomial of a fixed degree, the quality of the approximation
depends crucially on the norm of $H$ --- see \Ref{ref:Arad-AL13} for
more details. A natural way to achieve this is by truncating all the
energy levels of the Hamiltonian \emph{outside} the interesting
region at some energy scale $\tau$. For consistency reasons, we
denote the ``interesting region'', which we wish to keep local, by
$\Lout$, and the region whose energies are to be truncated by
$\Lin$. The exact definition of $\tH$ is then
\begin{deff}[The truncated Hamiltonian $\tH$]
\label{def:Ht} 
  Let $L$ be a subset of particles with its associated decomposition
  $H = \Hin + \Hb + \Hout $ as in \Sec{sec:settings}, and let
  $\tau>0$ be some fixed energy truncation scale. Let $P_{[0,\tau)},
  P_{[\tau,\infty)}$ be spectral projections associated with $\Hin$.
  Then the truncation of $\Hin$ is the Hamiltonian
  \begin{align}
  \label{def:tHin}
    \tHin \EqDef \Hin P_{[0,\tau)} + \tau P_{[\tau,\infty)} \,,
  \end{align}
  and the truncation of $H$ (with respect to $L$) is the Hamiltonian
  \begin{align}
  \label{def:tH}
    \tH \EqDef \tHin + \Hb + \Hout \,.
  \end{align}
  
  Eigenstates of $\tH$ will be denoted by $\ket{\tilde{\psi}_0},
  \ket{\tilde{\psi}_1}, \ket{\tilde{\psi}_2},\ldots$, and their
  corresponding energy levels by $\tilde{\eps}_0
  \le\tilde{\eps}_1\le\tilde{\eps}_2\le \ldots$. We also denote a
  projection into the subspace of eigenstates of $\tH$ with energies
  in the range $I$ by $\tilde{\Pi}_I$.
\end{deff}
We note that the norm of the truncated Hamiltonian $\tH$ is bounded
by $\norm{\tH} \le |\Lout| + |\Lb| + \tau$, so if $\Lout$ and
$\tau$ are of constant size, then so is $\norm{\tH}$. In what
follows, we shall always assume that $\tau$ is a fixed constant.

This definition of the truncated Hamiltonian would only be useful 
if $\tH$ is a good approximation to $H$, at least for the
lower parts of the spectrum. The following theorem uses
\Thm{thm:dist} and \Thm{thm:expE} to prove that this is indeed the
case: the lower part of the spectrum of $H$ and $\tH$ are
exponentially close to each other in $\tau$.
\begin{theorem}
\label{thm:lowspec} 
  The low energy subspaces and spectrum of $H$ and $\tH$ 
  are exponentially close in the following sense:
  \begin{description}
    \item [(i)]
      \begin{align}
      \label{eq:Pi-bound}
        \norm{(H-\tH)\Pi_{[0,\eps]}} 
          &\le \frac{6}{\lambda^{3/2}} 
            e^{-\lambda(\Delta\tau-\Delta\eps-3|\Lb|)} \,,
      \end{align}
      and
      \begin{align}
      \label{eq:tPi-bound}
        \norm{(H-\tH)\tPi_{[0,\eps]}} 
          &\le \frac{6}{\lambda^{3/2}} 
            e^{-\lambda(\Delta\tau-\Delta\teps-33|\Lb|)} \,,
      \end{align}
      where $\Delta \eps \EqDef \eps-\Egs$,
      $\Delta\teps\EqDef \eps-\tEgs$, and $\Delta\tau\EqDef \tau-\Ein$.
      
    \item  [(ii)] If $\eps_0 \leq \eps_1 \leq \eps_2 \dots$ 
      (respectively $\teps_0 \leq \teps_1 \leq \teps_2 \dots$) are
      the list of eigenvalues of $H$ (respectively $\tH$) in
      increasing order (with multiplicity) then for $\eps_j\le\eps$
      \begin{align} 
      \label{e:spectrum}
         \eps_j - \frac{6}{\lambda^{3/2}} 
          e^{-\lambda(\Delta\tau-\Delta\teps-33|\Lb|)} 
            \le \teps_j \le \eps_j \,.
      \end{align}   
  \end{description}
  
\end{theorem}


\section{Proof of \Thm{thm:expE}}
\label{sec:expE}

In this section we prove \Thm{thm:expE}, which serves as the
technical basis for all other theorems. Following ideas from an
unpublished result by Hastings, we fix some constant $s>0$, and write
\begin{align}
\label{eq:Hastings}
  \norm{\Pi_{[\eps',\infty)}\, A \, \Pi_{[0,\eps]}} =
  \norm{\Pi_{[\eps',\infty)}\, e^{-s H}e^{s H}\,
    A\, e^{-s H}e^{s H}\, \Pi_{[0,\eps]}}
  \le  \norm{e^{s H}Ae^{-s H}}\cdot 
    e^{-s (\eps'-\eps)} \,.
\end{align}
Our task is then to bound $\norm{e^{sH}Ae^{-sH}}$ and then find the
$s$ that minimizes the product $\norm{e^{s H}Ae^{-s H}}\cdot e^{-s
(\eps'-\eps)}$. We begin with bounding $\norm{e^{s H}Ae^{-s H}}$:
\begin{lem}
\label{lem:Hadamard-bound} 
  For any $0\le s<\frac{1}{gk}$ we have
  $\norm{e^{sH}Ae^{-sH}} \le \norm{A}\cdot(1-gks)^{-R/gk}$.
\end{lem}
\begin{proof}
  Without loss of generality, we can assume that $\norm{A}=1$, since
  a simple scaling of the equations proves the general result. Using
  the Hadamard formula (see, for example, Lemma~5.3, pp~160 in
  \Ref{ref:Miller72}), we write
  \begin{align}
    e^{s H}Ae^{-s H} 
    &= A + s[H,A] + \frac{s^2}{2!}[H,[H,A]] + \ldots
    \EqDef \sum_{\ell=0}^\infty \frac{s^\ell}{\ell!}K_\ell
    \label{eq:expansion} \,.
  \end{align}
  Then $\norm{e^{s H}Ae^{-s H}} \le \sum_{\ell=0}^\infty
  \frac{s^\ell}{\ell!}\norm{K_\ell}$. We shall upper bound the norm
  of each $K_\ell$ separately. Clearly, for $\ell=0$,
  $\norm{K_\ell}=\norm{A}=1$. For the $\ell>0$ case, we write
  $K_\ell$ as
  \begin{align}
  \label{eq:Kell}
    K_\ell = \sum_{X_1\in E_A} \ \ \sum_{X_2|X_1}
      \ \ \sum_{X_3|(X_2,X_1)}\cdots
     \sum_{X_\ell|(X_{\ell-1},\ldots,X_1)}
     [h_{X_\ell},[h_{X_{\ell-1}},\ldots, [h_{X_1},A]\ldots]]
  \end{align}
  Above, the sum $\sum_{X_j|X_{j-1},\ldots, X_1}$ denotes a
  summation over the $X_j$ subsets for which the commutator
  $[h_{X_j},[h_{X_{j-1}},[\ldots, [h_{X_1}, A]\ldots]]$ is non-zero.
  By assumption, for the first level $[H,A]$, we can take only
  $X_1\in E_A$. Once $X_1$ is fixed, then for the next level
  $[h_{X_2},[h_{X_1},A]]$ we can take $X_2$ which is either in $E_A$
  or which does not commute with $X_1$, and so on and so forth.

  To upperbound the norm of $K_\ell$ we note that for every operator
  $O$, $\norm{[h_X,O]}\le \norm{h_X}\cdot\norm{O}$. This is because
  we may define $\tilde{h}_X\EqDef h_X - \frac{1}{2}\norm{h_X}$, and
  using the fact that $h_X$ is a non-negative operator, it
  follows that $\norm{\tilde{h}_X}\le \frac{1}{2}\norm{h_X}$ and so
  $\norm{[h_X,O]} = \norm{[\tilde{h}_X,O]} \le
  2\norm{\tilde{h}_X}\cdot\norm{O} \le \norm{h_X}\cdot\norm{O}$
  Taking the norm of \Eq{eq:Kell} and using the fact that
  $\norm{A}=1$, we get
  \begin{align*}
    \norm{K_\ell} &\le \sum_{X_1\in E_A} 
      \ \ \sum_{X_2|X_1}\ \ \sum_{X_3|(X_1,X_2)}\cdots
     \sum_{X_\ell|(X_{\ell-1},\ldots,X_1)}
      \norm{h_{X_1}}\cdots\norm{h_{X_\ell}} \\
    &= \sum_{X_1\in E_A} \norm{h_{X_1}} \sum_{X_2|X_1} \norm{h_{X_2}}
      \sum_{X_3|(X_1,X_2)}\norm{h_{X_3}}\cdots
       \sum_{X_\ell|(X_{\ell-1},\ldots,X_1)} \norm{h_{X_\ell}}\,.
  \end{align*}
  Let us now upperbound the sums. The sum over $h_{X_\ell}$ includes
  only terms that do not commute with either $A$ or one of $h_{X_1},
  \ldots, h_{X_{\ell-1}}$. By assumption, the sum of the norms of
  $h_X$ that do not commute with $A$ is bounded by $R$. The sum of
  norms of $h_X$ that do not commute with another $h_Y$ is bounded
  by $gk$ since $h_Y$ is supported on at most $k$ particles. We
  therefore conclude that 
  \begin{align*}
    \sum_{X_\ell|(X_{\ell-1},\ldots,X_1)} \norm{h_{X_\ell}} \le
    R+(\ell-1)gk \,.
  \end{align*}
  Similarly, for any $1\le j\le \ell$, we get
  \begin{align*}
    \sum_{X_j|(X_{j-1},\ldots,X_1)} \norm{h_{X_j}} \le R+(j-1)gk \,,
  \end{align*}
  and therefore
  \begin{align*}
    \norm{K_\ell} &\le R(R  + gk) \cdot (R+2gk)
      \cdots \big(R + (\ell-1) gk\big) \\
    &=(gk)^\ell r(r+1)\cdots(r+\ell-1) \,,
  \end{align*} 
  where we defined $r \EqDef \frac{R}{gk}$. Plugging this into
  \Eq{eq:expansion} gives 
  \begin{align*}
     \norm{e^{s H}Ae^{-s H}}
       &\le \sum_{\ell=0}^\infty 
         \frac{(s gk)^\ell}{\ell!}
         r(r+1)\cdots(r+\ell-1) 
       = \frac{1}{(1-s gk)^r}\,,
  \end{align*} 
  where the last equality follows from a simple Taylor expansion,
  and is valid as long as $0<1-s gk\le 1$. 
\end{proof}
Moving on, \Lem{lem:Hadamard-bound} together with \Eq{eq:Hastings}
imply
\begin{align}
\label{eq:min-s}
  \norm{\Pi_{[\eps',\infty)}\, A \, \Pi_{[0,\eps]}}
    \le \frac{e^{-s (\eps'-\eps)}}{(1-s gk)^{R/gk}}\cdot\norm{A} \,.
\end{align}  
To finish the proof we now look for $0\le s<1$ that minimizes the
RHS above. A simple calculus shows that we should pick
$s=\frac{1}{gk}\left[1-\frac{R}{\eps'-\eps}\right]$, and 
substituting it \eqref{eq:min-s} proves the first inequality in
\eqref{eq:general-expE}. To prove the second inequality, we rewrite
the expression in the exponent as
\begin{align*}
  -\lambda(\eps'-\eps-2R) 
    -\lambda\left[\eps-\eps'
      -2R\ln\Big(\frac{\eps'-\eps}{R}\Big)\right] \,,
\end{align*}
and notice that $\eps'-\eps-2R\ln\Big(\frac{\eps'-\eps}{R}\Big)\ge
0$ for every $\eps'-\eps>0$.

\section{Proof of \Thm{thm:dist}}
\label{sec:dist}

We begin with a simple lemma, which upperbounds the norm of any
state of the form $\ket{\phi} = A\Pi_{[0,\eps]}\ket{\psi}$ in terms
of its energy with respect to $H$.
\begin{lem}
\label{lem:normphi} 
  Under the same conditions of \Thm{thm:expE},
  let $\ket{\psi}$ be an arbitrary normalized state and define
  $\ket{\phi}\EqDef A\Pi_{[0,\eps]}\ket{\psi}$ and its energy
  $\eps_\phi \EqDef \frac{1}{\norm{\phi}^2}\bra{\phi}H\ket{\phi}$.
  Then,
  \begin{align}
  \label{eq:normphi}
    \norm{\phi} \le \norm{A}\cdot
      \frac{2}{\lambda^{1/2}} e^{-\lambda(\eps_\phi-\eps-2R)} \,,      
  \end{align}
  where $R$ is defined as in \Thm{thm:expE}.
\end{lem}

\begin{proof}
  As with the proof of \Thm{thm:expE}, we can assume without loss of
  generality that $\norm{A}=1$. Let $\mu$ be some energy scale to be
  set later, define $h\EqDef \frac{\ln 2}{2\lambda}$ and write
  \begin{align*}
    \ket{\phi} &= \Pi_{[0,\mu)}\ket{\phi} + \sum_{j=0}^\infty 
     \Pi_{[\mu+jh,\mu+(j+1)h)}\ket{\phi} 
     \EqDef \ket{\phi_{-1}}  
     + \sum_{j=0}^\infty\ket{\phi_j} \,.
  \end{align*}
  \Thm{thm:expE} establishes that the norms of the $\ket{\phi_{j}}$
  decay exponentially, i.e.,
  \begin{align}
  \label{eq:z6}
    \norm{\phi_j}^2  
      = \norm{\Pi_{[\mu+jh,\mu+(j+1)h)}A\Pi_{[0,\eps]}\ket{\psi}}^2 
      &\le \norm{\Pi_{[\mu+jh,\infty)}A\Pi_{[0,\eps]}\ket{\psi}}^2  
       \le e^{-2\lambda(\mu+jh-\eps-2R)} \,. 
  \end{align}
  We use this decomposition to bound the energy of $\ket{\phi}$ with
  respect to $H$:
  \begin{align} 
  \label{eq:z5b}
    \bra{\phi}H\ket{\phi} &=  \bra{\phi_{-1}}H\ket{\phi_{-1}} 
      + \sum_{j=0}^\infty  \bra{\phi_j}H\ket{\phi_j} \\
    &\le \mu \norm{\phi_{-1}}^2 
      + \sum_{j=0}^{\infty} ( \mu + (j+1)h) \norm{\phi_j}^2 
      \nonumber \\
    &\le \mu \norm{\phi}^2 
      + h\sum_{j=0}^{\infty} (j+1) \norm{\phi_j}^2 \,.
      \nonumber
  \end{align}
  We bound the rightmost sum using \eqref{eq:z6}:
  \begin{align}
  \label{eq:z7}      
    \sum_{j=0}^{\infty} (j+1) \norm{\phi_j}^2 
      &\le  e^{-2\lambda(\mu-\eps-2R)}
        \sum_{j=0}^\infty (j+1) e^{-2\lambda h j}  
      = e^{-2\lambda(\mu-\eps-2R)} \sum_{j=0}^\infty (j+1) 2^{-j} \,.
  \end{align} 
  The final summand in \eqref{eq:z7} is equal to 4 by a standard
  equality; combining this with \eqref{eq:z5b} yields the bound of
  the energy as:
  \begin{align}
    \bra{\phi}H\ket{\phi}=\eps_\phi\norm{\phi}^2 
      &\le \mu\norm{\phi}^2 
        + 4he^{-2\lambda(\mu-\eps-2R)} \,.
  \end{align}
  Choosing $\mu\EqDef \eps_\phi-1$, rearranging terms and taking a
  square root,  we get
  \begin{align*}
      \norm{\phi} 
        \le (4h)^{1/2}e^{\lambda}
        e^{-\lambda(\eps_\phi-\eps-2R)} 
      \le \frac{2}{\lambda^{1/2}} e^{-\lambda(\eps_\phi-\eps-2R)}\,,
  \end{align*}
  where the last inequality follows from the fact that $\lambda\le
  \frac{1}{2}$ and so $(4h)^{1/2}e^{\lambda}\le 2/\lambda^{1/2}$.
  This proves \eqref{eq:normphi} for $\norm{A}=1$.
\end{proof}

The proof of \Thm{thm:dist} will follow by applying
\Lem{lem:normphi} with $A=P_{[\tau,\infty)}$. In this case, $[A,\Hin] =
[A,\Hout]=0$, so the only non-commuting terms in $[A,H]$ come from
$\Hb$ and thus we can take $R=|\Lb|$ and
\begin{align}
\label{eq:normphi2}
  \norm{\phi} \le \frac{2}{\lambda^{1/2}}
    \cdot e^{-\lambda(\eps_\phi-\eps-2|\Lb|)} \,.
\end{align}
We now lowerbound $\eps_\phi$. By definition,
\begin{align}
  \eps_\phi &= \frac{1}{\norm{\phi}^2}\bra{\phi}\Hin\ket{\phi}
    +\frac{1}{\norm{\phi}^2}\bra{\phi}\Hb\ket{\phi}
    +\frac{1}{\norm{\phi}^2}\bra{\phi}\Hout\ket{\phi}
    \ge \tau + \Eout \,. \label{eq:z7b}
\end{align}
We can further lower bound the right hand side by noting that
$\eps_0 \le \Ein + |\Lb| + \Eout$,\footnote{This follows from
bounding the energy of the total Hamiltonian $H$ with respect to the
product state $\ket{\psi_0(\Lin)}\otimes\ket{\psi_0(\Lout)}$, where
$\ket{\psi_0(\Lin)}$ and $\ket{\psi_0(\Lout)}$ the groundstates of
$\Hin$ and $\Hout$ respectively. On one hand, it must be
lowerbounded by $\Egs$, the groundenergy of $H$, and on the other
hand, it must be upperbounded by $\Ein+\Eout+|\Lb|$ since by the
definition of $|\Lb|$ in \Eq{def:sizes}, the norm of $\Hb$ is
upperbounded by $|\Lb|$.} and therefore $\Eout \ge
\Egs-\Ein-|\Lb|$.  Using this in \eqref{eq:z7b} gives 
\begin{align*}
  \eps_\phi \ge \eps_0 + \tau-\Ein -|\Lb|  
    = \eps_0 + \Delta\tau - |\Lb|\,,
\end{align*}
and substituting this in \eqref{eq:normphi2} yields
\eqref{eq:EnergyDist}.

\section{Proof of \Thm{thm:product}}
\label{sec:product}

The proof of \Thm{thm:product} is very similar to that of
\Thm{thm:dist}; we will only give the outline of the proof and
highlight where things are different.
 
To prove \eqref{eq:t-gt-eps}, note that \Thm{thm:dist} holds in the
slightly modified context (the proof is identical) of replacing the
Hamiltonian $H_L$ with the Hamiltonian $H_L + H_{L^c}$ and replacing
$P_{[\tau,\infty)}$ with $Q_{[\tau,\infty)}$. Therefore, it implies:
\begin{align*} 
  \norm{Q_{[\tau,\infty)}\Pi_{[0,\eps]}} 
    \le \frac{2}{\lambda^{1/2}} 
      e^{-\lambda(\Delta\tau-\Delta \eps - 3|\Lb|)} \,,
\end{align*}
where, $\Delta\tau \EqDef \tau-(\Ein+\Eout)$ and $\Delta\eps \EqDef
\eps - \Egs$. Expanding the terms in the exponential, we get
$\Delta\tau-\Delta \eps - 3|\Lb| = [\Ein+\Eout-\Egs] + \tau-\eps -
3|\Lb|$, and as $\Ein+\Eout \le \Egs$,\footnote{This is because
if $\ket{\psi_0}$ is the groundstate of $H$ then $\Ein+\Eout\le
\bra{\psi_0}\Hin\ket{\psi_0} + \bra{\psi_0}\Hout\ket{\psi_0} \le
\bra{\psi_0}\Hin\ket{\psi_0} + \bra{\psi_0}\Hout\ket{\psi_0} +
\bra{\psi_0}\Hb\ket{\psi_0} = \bra{\psi_0}H\ket{\psi_0} = \Egs$.}
it follows that $e^{-\lambda(\Delta\tau-\Delta \eps - 3|\Lb|)}\le 
e^{-\lambda(\tau-\eps - 3|\Lb|)}$, and therefore
\begin{align*} 
  \norm{Q_{[\tau,\infty)}\Pi_{[0,\eps]}} 
    \le \frac{2}{\lambda^{1/2}}
      e^{-\lambda(\tau-\eps - 3|\Lb|)} \,.
\end{align*}
The final step in proving \eqref{eq:t-gt-eps}, is to use the
identities $\norm{Q_{[\tau,\infty)}\Pi_{[0,\eps]}} =
\norm{(Q_{[\tau,\infty)}\Pi_{[0,\eps]})^\dagger} =
\norm{\Pi_{[0,\eps]}Q_{[\tau,\infty)}}$.

To prove \eqref{eq:eps-gt-t}, we first view it as a
``complementary'' version of \eqref{eq:t-gt-eps}, in which the roles
of the main Hamiltonian $H$ (with corresponding spectral operator
$\Pi_{[0,\eps]}$) and the partial Hamiltonian $\Hin+\Hout$ (with
corresponding spectral operator $Q_{[\tau, \infty)}$) have been
switched.  Therefore, to prove it, we shall need the following
``complementary" version of \Thm{thm:expE}:
\begin{theorem}
\label{thm:expE2}
  Under the same conditions of \Thm{thm:product}, let $A$ be an
  operator that commutes with $H$. Then
  \begin{align}
    \norm{Q_{[\tau',\infty)} A Q_{[0,\tau]} }
      \le e^{-\lambda(\tau'-\tau-2|\Lb|)} \cdot \norm{A}\,.
  \end{align}
\end{theorem}
\begin{proof}
  The proof here is exactly like the proof of \Thm{thm:expE}, and so
  we leave it as an exercise to the reader.
\end{proof}

With \Thm{thm:expE2} at our disposal, we use the same argument as in
\Lem{lem:normphi} to deduce that for every
$\ket{\phi}=\Pi_{[\eps,\infty)}Q_{[0,\tau]}\ket{\psi}$, 
\begin{align*}
  \norm{\phi} \le \frac{2}{\lambda^{1/2}}
  e^{-\lambda(\eps_\phi-\tau-3|\Lb|)} \,,
\end{align*}
where $\eps_\phi$ is the energy of $\ket{\phi}$. We complete the
proof by lower bounding $\eps_\phi$, the energy of $\ket{\phi}$.
Since $\Hin+\Hout=H-\Hb\ge H-|\Lb|$, we conclude that $\eps_\phi
\ge\eps-|\Lb|$, which gives $\norm{\phi} \le \frac{2}{\lambda^{1/2}}
e^{-\lambda(\eps-\tau-3|\Lb|)}$, thereby proving \eqref{eq:eps-gt-t}.

\section{Proof of \Thm{thm:lowspec}}
\label{sec:Ht}

We begin by proving part (i) of the theorem. Since \Thm{thm:dist}
says that the high energy spectrum of $H_L$ and the low energy
spectrum of $H$ have very little overlap, it is a natural tool for
bounding the left hand side of \eqref{eq:Pi-bound}, which can be
written as $(H-\tH)\Pi_{[0,\eps]}=(\Hin
-\tau)P_{[\tau,\infty)}\Pi_{[0,\eps]}$.  We decompose
$[\tau,\infty)= \sqcup_{j=0}^{\infty} I_j$ with $I_j\EqDef[\tau+jh,
\tau+(j+1)h)$, with $h\EqDef \frac{\ln{2}}{\lambda}$.  This allows
us to write $P_{[\tau,\infty)}= \sum_j P_{I_j}$ where $P_{I_j}$
are spectral projections associated to $\Hin$. Then by the triangle
inequality,
\begin{align*}
  \norm{(H-\tH)\Pi_{[0,\eps]}} 
    \le \sum_{j\ge 0}\norm{(\Hin-\tau)P_{I_j}\Pi_{[0,\eps]}}
    \le \sum_{j\ge 0}[\tau+(j+1)h-\tau]\cdot\norm{P_{I_j}\Pi_{[0,\eps]}}
    = h\sum_{j\ge 0}(j+1)\cdot\norm{P_{I_j}\Pi_{[0,\eps]}}\,.
\end{align*}  
Using \Thm{thm:dist} to bound each term in the summand, we have 
\begin{align*}
  \norm{P_{I_j}\Pi_{[0,\eps]}} 
    \le \frac{2}{\lambda^{1/2}}
      e^{-\lambda(\Delta\tau+jh-\Delta\eps - 3|\Lb|)} \,,
\end{align*}
and so 
\begin{align*}
  \norm{(H-\tH)\Pi_{[0,\eps]}} 
    \le \frac{2h}{\lambda^{1/2}}
      e^{-\lambda(\Delta\tau-\Delta\eps - 3|\Lb|)}
        \sum_{j\ge 0}(j+1)e^{-\lambda h j} \,.
\end{align*}  
Since $e^{-\lambda h j}=\left(\frac{1}{2}\right)^j$, then by the
identity $\sum_{j\ge 0}(j+1)2^{-j} = 4$, the RHS becomes
$\frac{8\ln 2}{\lambda^{3/2}} e^{-\lambda(\Delta\tau-\Delta\eps -
3|\Lb|)}$, and as $8\ln 2 \le 6$, we recover \eqref{eq:Pi-bound}. 

For the proof of \eqref{eq:tPi-bound} we first need an analogous
statement to \Thm{thm:dist}, which says that the overlap of the high
energy spectrum of $H_L$ and the low energy spectrum of $\tH$ has
very little overlap:
\begin{theorem}
  \label{thm:dist2}
  Let $P_{[\tau,\infty)}$ denote the projection onto the
  subspace of energies of $H_L$ which are $\ge \tau$, and let
  $\tPi_{[0,\eps]}$ denote the projection onto the subspace of
  energies $\tH$ that are $\le \eps$. Then
  \begin{align}
  \label{eq:EnergyDist2}
    \norm{P_{[\tau,\infty)}\tPi_{[0,\eps]}} \le 
    \frac{2}{\lambda^{1/2}}
      \cdot e^{-\lambda(\Delta \tau - \Delta \teps - 33|\Lb|)} \,,
  \end{align}
  where $\Delta\tau \EqDef \tau-\Ein$ and $\Delta\teps \EqDef
  \eps-\teps_0$.
\end{theorem}
The proof is given in the next subsection. With this result in hand,
the proof of \eqref{eq:tPi-bound} follows the identical route as
\eqref{eq:Pi-bound} above with \Thm{thm:dist2} replacing
\Thm{thm:dist}, and adjusting the boundary term from $3|\Lb|$ to
$33|\Lb|$.

For (ii), since $\tH\le H$ as operators, it follows immediately that
for every $j$, $\teps_j\le\eps_j$.\footnote{This is an immediate
consequence of Weyl's inequality for matrices. See, for example,
\Ref{ref:Franklin2012matrix}, pp.~157.} For the other inequality,
recall a useful fact about the $j^{th}$ smallest eigenvalue
$\lambda_j$ of a self-adjoint operator $A$: for any projector $P$ of
rank $j$, 
\begin{align} 
\label{e:minnorm}
  \lambda_j\le \norm{PAP} \,,
\end{align}
with equality when $P$ is chosen to be the projector onto the span
of the lowest $j$ eigenvectors of $A$.  Setting $\tP$ to be the
projector onto the span of the lowest $j$ eigenvectors of $\tH$
yields $\norm{\tP\tH \tP}=\teps_j$.  Then by the triangle inequality
\begin{align}
\label{eq:teps-j-bound}
  \teps_j= \norm{\tP\tH \tP} \ge \norm{\tP H \tP} 
    - \norm{\tP(\tH-H)\tP}
  \ge \eps_j - \norm{\tP(\tH-H)\tP} \,.
\end{align}
To upperbound $\norm{\tP(\tH-H)\tP}$, we use inequality
\eqref{eq:tPi-bound} of part (i), which implies
$\norm{\tP(\tH-H)\tP}\le\frac{6}{\lambda^{3/2}}
e^{-\lambda(\Delta\tau -\Delta\teps_j-33|\Lb|)}$. As $\teps_j
\le \eps_j\le \eps$, it follows that $\Delta\teps_j\le \Delta\teps$,
and so 
\begin{align*}
  \norm{\tP(\tH-H)\tP}\le \frac{6}{\lambda^{3/2}}
    e^{-\lambda(\Delta\tau-\Delta\teps-33|\Lb|)} \,.  
\end{align*}
Substituting this in \eqref{eq:teps-j-bound} finishes the proof.

We now move to the proof of \Thm{thm:dist2}.

\subsection{Proving \Thm{thm:dist2}}

The proof of \Thm{thm:dist2} follows closely that of
\Thm{thm:dist}. Looking at that proof, it is easy
to see that it generalizes to $\tH$, \emph{provided} we have a
version of \Thm{thm:expE} that applies to projectors of $\tH$
(instead of $H$) and an operator $A=P_{[0,\infty)}$. Given such a
theorem, all that is left to do is to adjust the prefactor in front
of the exponent, which we leave for the reader. We shall therefore
concentrate on proving the following version of \Thm{thm:expE}:
\begin{lem}
\label{lem:expE2} 
  Let $A$ be an operator that is supported by a subset of particles
  $L$, and assume that it commutes with $H_L$.  Then
  \begin{align}
  \label{eq:trucated-dist}
    \norm{\tPi_{[\eps',\infty)} A \tPi_{[0,\eps]}}
      \le \norm{A}\cdot e^{-\lambda(\eps'-\eps-32|\Lb|)} \,.
  \end{align}
\end{lem}
\begin{proof}[\ of \Lem{lem:expE2}]
  As in the proof of \Thm{thm:expE}, assume without loss of
  generality that $\norm{A}=1$, and insert
  $e^{-\lambda \tH}e^{\lambda \tH}$ before and after $A$ in the LHS
  of \eqref{eq:trucated-dist}. We get,
  \begin{align}
    \norm{\tPi_{[\eps',\infty]}A \tPi_{(0,\eps]}} 
     \le e^{-\lambda(\eps'-\eps)} \cdot
       \norm{e^{\lambda \tH} A e^{-\lambda \tH}} \,.
  \end{align}
  Our goal is then to show that $\norm{e^{\lambda \tH} A e^{-\lambda
  \tH}}\le e^{32\lambda |\Lb|}$. However,
  since $\tH$ contains non-local terms, we can no longer prove this
  using the Hadamard formula, as we did in the proof of \Thm{thm:expE}.
  As an alternative approach, we use the Dyson expansion:
  \begin{lem}[ Dyson expansion]
  \label{lem:Dyson}
    For any two operators $X,Y$ and a real number $t\ge 0$,
    \begin{align}
      \label{eq:Dyson}
      e^{t(X+Y)} &= \sum_{{j}=0}^\infty G_j(t)e^{tX} ,\quad \text{and}
      \quad
      e^{-t(X+Y)} = e^{-tX}\sum_{{j}=0}^\infty G'_j(t)
    \end{align}
    where, 
    \begin{align}      
      G_j(t)&\EqDef \int_0^t \!\! ds_1\int_0^{s_1}\!\!ds_2
        \cdots \int_0^{s_{j-1}} \!\!\!\! ds_j \, 
        Y(s_j)\cdots  Y(s_2)\cdot Y(s_1) \,,\\  
      G'_j(t)&\EqDef (-1)^{j}\int_0^t \!\! ds_1\int_0^{s_1}\!\!ds_2
        \cdots \int_0^{s_{j-1}} \!\!\!\! ds_j \, Y(s_1)\cdot
        Y(s_2)\cdots Y(s_j) \,,\\
        Y(s) &\EqDef e^{s X} Y e^{-s X} \,,
    \end{align}
    and $G_0(t) = G'_0(t) = \Id$.
  \end{lem}
  The proof is given in the appendix. 
  
  Recalling that $\tH = \tHin + \Hb + \Hout$, we let $X\EqDef\tHin +
  \Hout$ and $Y\EqDef\Hb$. Then
  \begin{align*}
    e^{\lambda \tH} A e^{-\lambda\tH}
     &= \sum_{j=0}^\infty \sum_{{j}'=0}^\infty 
       G_j(\lambda) e^{\lambda(\tHin + \Hout)}
       A e^{-\lambda(\tHin + \Hout)}G'_{{j}'}(\lambda) \\
     &=  \sum_{{j}=0}^\infty \sum_{{j}'=0}^\infty
       G_j(\lambda) A G'_{{j}'}(\lambda) \,,
  \end{align*}
  where in the last equality we used the fact that $A$ commutes with
  $\Hin$ and is supported on $L$, and so it also commutes with 
  $\tHin + \Hout$. By the triangle inequality, it follows that
  \begin{align}
  \label{eq:sum-norm}
    \norm{e^{\lambda \tH} A e^{-\lambda\tH}}
      \le \sum_{{j}=0}^\infty \sum_{{j}'=0}^\infty
         \norm{G_j(\lambda)}\cdot\norm{G'_{{j}'}(\lambda)} 
    = \left(\sum_{{j}=0}^\infty \norm{G_j(\lambda)}\right)
         \cdot\left(\sum_{{j}'=0}^\infty
       \cdot\norm{G'_{{j}'}(\lambda)}\right) \,.
  \end{align}
  
  Our task is then to bound $\norm{G_j(\lambda)}$ and
  $\norm{G'_{{j}'}(\lambda)}$. Using the definition of the Dyson
  expansion in \Lem{lem:Dyson}, we have 
  \begin{align*}      
    G_j(\lambda)&\EqDef \int_0^\lambda \!\! ds_1\int_0^{s_1}\!\!ds_2
      \cdots \int_0^{s_{{j}-1}} \!\!\!\! ds_j \,
      \Hb(s_j)\cdots  \Hb(s_2)\cdot \Hb(s_1) \,,\\
    G'_j(\lambda)&\EqDef (-1)^{j}\int_0^\lambda 
      \!\! ds_1\int_0^{s_1}\!\!ds_2
      \cdots \int_0^{s_{{j}-1}} \!\!\!\! ds_j \, 
      \Hb(s_1)\cdot \Hb(s_2)\cdots \Hb(s_j) \,,
  \end{align*}
  where $\Hb(s)\EqDef e^{s (\tHin + \Hout)} \Hb e^{-s (\tHin
  + \Hout)}$. To proceed, we need the following lemma, which is
  proved by the end of this section.
  \begin{lem}
  \label{lem:Hb-bound}
    $\norm{\Hb(s)} \le 16|\Lb|$ for all $0\le
    s\le \lambda$.
  \end{lem}
  Defining $c\EqDef 16|\Lb|$, we can use the lemma to bound
  $\norm{G_j(\lambda)}$:
  \begin{align*}
    \norm{G_j(\lambda)} &\le \int_0^\lambda \!\! 
      ds_1\int_0^{s_1}\!\!ds_2
      \cdots \int_0^{s_{{j}-1}} \!\!\!\! ds_j \,
      \norm{\Hb(s_j)}\cdots  \norm{\Hb(s_1)} \,,\\
    &= \frac{1}{{j}!} \left(\int_0^\lambda \norm{\Hb(s)}\right)^{j}
    \le \frac{1}{{j}!} (\lambda c)^{j} \,.
  \end{align*}
  Similarly, $\norm{G'_{{j}'}(\lambda)}\le \frac{1}{{j}!} (\lambda
  c)^n$. Therefore, $\sum_{{j}=0}^\infty \norm{G_j(\lambda)}\le
  e^{\lambda c}$ and $\sum_{{j}'=0}^\infty\norm{G'_{{j}'}(\lambda)} \le
  e^{\lambda c}$, which, upon substitution in \eqref{eq:sum-norm},
  proves that
  \begin{align*}
    \norm{e^{\lambda \tH} A e^{-\lambda\tH}}
      \le e^{2\lambda c} = e^{32\lambda |\Lb|} \,.
  \end{align*}
  
  We finish the proof by proving \Lem{lem:Hb-bound}.
  \begin{proof}[\ of \Lem{lem:Hb-bound}]
    We will show that for every $X\in E_{\Lb}$,
    $\norm{e^{s(\tHin+\Hout)} h_X e^{-s(\tHin+\Hout)}} \le
    16\norm{h_X}$, from which it follows that
    \begin{align*}
      \norm{e^{s(\tHin+\Hout)} \Hb e^{-s(\tHin+\Hout)}}
        \le 16\sum_{X\in E_{\Lb}} \norm{h_X} = 16|\Lb| \,.
    \end{align*}
    
    Since $[\tHin,\Hout]=0$, we can write 
    \begin{align*}
      e^{s(\tHin+\Hout)} h_X e^{-s(\tHin+\Hout)}
        = e^{s\tHin} O e^{-s\tHin} \,,
    \end{align*}
    where $O\EqDef e^{s\Hout} h_X e^{-s\Hout}$. We first apply
    \Lem{lem:Hadamard-bound} to bound $\norm{O}$, by noting that for
    $A=h_X$ we can use $R=gk$, and consequently, $\norm{e^{s\Hout}
    h_X e^{-s\Hout}} \le (1-sgk)^{-1}\cdot\norm{h_X}$. Since $s\le
    \lambda=1/(2gk)$, it follows that $\norm{O}\le 2\norm{h_X}$.
    
    Next, we wish to bound $\norm{e^{s\tHin}Oe^{-s\tHin}}$. To
    this aim, let us bound the norm of 
    $\ket{\phi}\EqDef{e^{s\tHin}O e^{-s\tHin}\ket{\psi}}$,
    where $\ket{\psi}$ is an arbitrary normalized state. For
    brevity, define $P_+ \EqDef P_{[\tau, \infty)}$, and $P_-\EqDef
    P_{[0,\tau)}$, where $P_{[\tau, \infty)}$ and $P_{[0,\tau)}$ are
    the projectors used in the definition of $\tH$ in
    \Def{def:Ht}. Then writing $\ket{\phi_{\pm\pm}}\EqDef P_\pm e^{s\tHin}O
    e^{-s\tHin}P_\pm\ket{\psi}$, we have
    \begin{align*}
      \ket{\phi} = \ket{\phi_{++}} +\ket{\phi_{+-}}
        + \ket{\phi_{-+}} +\ket{\phi_{--}} \,.
    \end{align*}
    We now bound the norm of each
    component separately using the fact that $P_+ e^{\pm s\tH} =
    P_+e^{\pm s\tau}$, and $P_- e^{\pm s\tHin} = P_-e^{\pm s
    \Hin}$.
    \begin{description}
      \item [$\ket{\phi_{++}}$:]
        By definition, $\ket{\phi_{++}} = P_+ e^{s \tau}O
        e^{-s \tau}P_+\ket{\psi} = P_+ OP_+\ket{\psi}$ and so
        $\norm{\phi_{++}} \le \norm{O}\cdot
        \norm{P_+\ket{\psi}}\le 2 \norm{h_X}\cdot\norm{P_+\ket{\psi}}$.

      \item [$\ket{\phi_{-+}}$:] 
        Here, $\ket{\phi_{-+}} = P_- e^{s\Hin}O e^{-s\tau}
        P_+\ket{\psi}$ and so $\norm{\phi_{-+}} \le e^{-s\tau}
        \cdot\norm{P_- e^{s\Hin}}\cdot
        \norm{O}\cdot\norm{P_+\ket{\psi}}$. But as $\norm{P_- e^{s
        \Hin}}\le e^{s\tau}$, we conclude that $\norm{\phi_{-+}}\le
        \norm{O}\cdot \norm{P_+\ket{\psi}}\le
        2\norm{h_X}\cdot\norm{P_+\ket{\psi}}$.
        
      \item [$\ket{\phi_{--}}$:] 
        Here $\ket{\phi_{--}} = P_-e^{s \Hin}O e^{-s
        \Hin}P_-\ket{\psi}$ so $\norm{\phi_{--}} \le
        \norm{e^{s\Hin}O e^{-s \Hin}}\cdot\norm{P_-\ket{\psi}}
        \le 2\norm{h_X}\cdot\norm{P_-\ket{\psi}}$, where we invoked
        \Lem{lem:Hadamard-bound} to deduce that $\norm{e^{s\Hin}O
        e^{-s \Hin}}\le 2\norm{h_X}$.\footnote{\label{fn:R} Note
        that just like when $A=h_X$, when $A=O=e^{s\Hout} h_X
        e^{-s\Hout}$, we can still use $R=gk$, because the
        operators $e^{\pm s \Hout}$ commute with the local terms
        of $\Hin$.}
      
      \item [$\ket{\phi_{+-}}$:] 
        We write $\ket{\phi_{+-}} = e^{s t}P_+ O e^{-s
        \Hin}P_-\ket{\psi}$. To bound its norm, we slice the energy
        range of $P_-$, i.e., $[0,\tau)$ into segments $I_j=[a_j,b_j)$
        of width $h\EqDef gk$, such that $I_0=[\tau-h,\tau),
        I_1=[\tau-2h,\tau-h)], \ldots$ (the last segment might be of
        shorter width). Then
        \begin{align*}
          \norm{\phi_{+-}} &\le e^{s t}\sum_{j\ge 0}
            \norm{P_+ O P_{I_j}}
            \cdot \norm{e^{-s \Hin}P_{I_j}\ket{\psi}} \,.
        \end{align*}
        As $\norm{P_+ O P_{I_j}} \le
        \norm{P_{[t,\infty)}OP_{[0,t-jh)}}$, we can use
        \Thm{thm:expE} with $R=gk$ (see footnote~\ref{fn:R}) to 
        obtain
        \begin{align*}
          \norm{P_+ O P_{I_j}} 
            \le \exp\left\{-\frac{1}{gk}
              \big[jh-gk(1+\ln(jh/gk))\big]\right\}
                \cdot\norm{O} 
            = j e^{-j+1}\cdot\norm{O}\,.
        \end{align*}
        In addition, $e^{s\tau}\norm{e^{-s\Hin}P_{I_j}\ket{\psi}} 
        \le e^{s\tau}e^{-s(\tau-jh-h)}\norm{P_{I_j}\ket{\psi}} 
        =e^{sjh+sh}\norm{P_{I_j}\ket{\psi}}$, and therefore, since
        $s\le \lambda$, we get $e^{s\tau}\norm{P_+ O P_{I_j}} \cdot
        \norm{e^{-s \Hin}P_{I_j}\ket{\psi}} \le \norm{O}\cdot
        e^{3/2} j e^{-j/2}\cdot\norm{P_{I_j}\ket{\psi}}$. Summing up
        the $j=0,1,2,\ldots$ contributions gives us
        \begin{align*}
          \norm{\phi_{+-}} &\le 
            \norm{O}\cdot e^{3/2}
              \sum_{j\ge 0} je^{-j/2}
                \cdot\norm{P_{I_j}\ket{\psi}} 
          \le \norm{O}\cdot e^{3/2}
              \left(\sum_{j\ge 0} j^2e^{-j}\right)^{1/2}
               \norm{P_-\ket{\psi}}\,,
        \end{align*}
        where the second inequality follows from the Cauchy-Schwartz
        inequality, together with the fact that $\sum_{j\ge 0}
        \norm{P_{I_j}\ket{\psi}}^2 = \norm{P_-\ket{\psi}}^2$.
        
        We now use $\norm{O}\le 2\norm{h_X}$, together with the 
        formula $\sum_{j\ge 0}j^2q^j=\frac{q(1+q)}{(1-q)^3}$ to get
        \begin{align*}
          \norm{\phi_{+-}} &\le
            \norm{h_X}\cdot 2e^{3/2}
            \cdot \left(\frac{e^{-1}(1+e^{-1})}{(1-e^{-1})^3}\right)^{1/2}
              \norm{P_-\ket{\psi}}
           \le 13\norm{h_X}\cdot\norm{P_-\ket{\psi}}\,.
        \end{align*}
    \end{description}
    All together, we find that $\norm{\phi} \le
    (4\norm{P_+\ket{\psi}} + 15\norm{P_-\ket{\psi}})
    \cdot\norm{h_X}$, so by invoking the Cauchy-Schwartz inequality
    once more, we get $\norm{\phi}\le \sqrt{4^2+15^2}\norm{h_X}\le
    16\norm{h_X}$.

  \end{proof}

\end{proof}

\section{Summary and future work}
\label{sec:summary}

In this paper we have rigorously proven several bounds on the local
and global energy distributions in quantum spin models on a lattice.
The common theme in all these results is that, to a large extent,
these energy distributions behave as if the underlying system is
commuting (or even classical), up to some exponentially small
corrections. Our bounds apply to a very wide family of systems: all
that is assumed is that the interactions are at most $k$-body and that the
total strength of the interactions that involve a particle is
finite. No other assumptions like nearest-neighbors interactions, 
spectral gap, shape of the spectrum, or the specific form of the
interactions is needed. Indeed, the most important ingredient that
was used is the fact that the system is made of many local
interactions, and that the influence of a single particle on the
total energy of the system is bounded by a constant. It is this
explicit locality that tames the quantum effects of
non-commutativity, and drives the system towards a more classical
behavior. 

The main motivation behind this paper was the need to construct a
good approximation for the ground state projector of a gapped system
(AGSP) using a low-degree polynomial of $H$. This was a central
building block of a recent 1D area-law proof \cite{ref:Arad-AL13}.
Nevertheless, since the results we have presented here are very
general, we hope that they might be useful at other places as well.
One example where our results have already been used is in a recent
result about the entanglement structure of gapped ground states,
known as ``Local reversibility''~\cite{Ref:LR}, which was published
after the first draft of this paper.

Finally, it is interesting to know how tight our bounds are.  This
can be studied by either optimizing our calculations, or by directly
estimating the energy distributions of particular examples, either
numerically or analytically, to see how they match our bounds. In
particular, some very simple numerical calculations, which we
performed on a chain of 12 spins with random interactions, suggest
that the energy distribution $\norm{\Pi_{[\eps', \infty)}
A\Pi_{[0,\eps]}}$ from \Thm{thm:expE} can be upperbounded by an
expression of the form $e^{-\orderof{ |\eps'-\eps-\orderof{R}| }
\log |\eps'-\eps-\orderof{R}| } \cdot \norm{A}$. It would be
interesting to see if such a stronger bound can also be proven
rigorously.  

\section*{Acknowledgments}

We are grateful to M.~B.~Hastings for sharing his unpublished proof
of the result upon which \Thm{thm:expE} is based. We thank
J.~I.~Latorre for useful discussions and comments on the manuscript.

Research at the Centre for Quantum Technologies is funded by the
Singapore Ministry of Education and the National Research
Foundation, also through the Tier 3 Grant ``Random numbers from
quantum processes''.
TK also acknowledges the support from the Program for Leading Graduate
Schools, MEXT, Japan and JSPS grant no. 2611111.


\bibliographystyle{ieeetr}

{~}

\bibliography{edist}

\appendix

\section{Proof of \Lem{lem:Dyson}}

We will only prove the first equality in \Eq{eq:Dyson}, i.e., 
$e^{t(X+Y)} = \sum_{{j}=0}^\infty G_j(t)e^{tX}$, as the proof of
second equality follows the exact same lines.
    
Define $L(t)\EqDef e^{t(X+Y)}$ and $R(t)\EqDef \sum_{{j}=0}^\infty
G_j(t)e^{tX}$, the LHS and RHS of the first equation in
\eqref{eq:Dyson} respectively. We wish to show that $L(t)= R(t)$ for
all $t\ge 0$.  We do that by showing that as a function of $t$, 
both satisfy the same linear ordinary differential equation with the
same initial condition. Indeed, at $t=0$, we have $L(0)=R(0)=\Id$.
Next, differentiating $L(t)$ gives us the equation $\frac{d}{dt}L(t)
= L(t)\cdot(X+Y)$. Let us show that the same holds for $R(t)$. By
definition, 
\begin{align*}
  \frac{d}{dt} R(t) = R(t)X 
    + \sum_{{j}=0}^\infty\frac{d}{dt}G_j(t)e^{tX} \,.
\end{align*}
But clearly $\frac{d}{dt}G_j(t) = G_{{j}-1}(t)Y(t)$ for ${j}>0$ and
is vanishing for ${j}=0$, and so
\begin{align*}
  \frac{d}{dt} R(t) = R(t)X
    + \sum_{{j}=0}^\infty G_j(t) Y(t) e^{tX}
    = R(t)X + \sum_{{j}=0}^\infty G_j(t)e^{tX} Y
    = R(t)\cdot (X+Y) \,,
\end{align*}
which concludes the proof. \qedsymb

\end{document}